\pdfoutput=1
\documentclass[11pt]{article}
\usepackage{amsmath, amssymb, amsthm, mathtools}
\usepackage{geometry}
\usepackage{hyperref}
\usepackage[numbers,sort&compress]{natbib}
\usepackage{setspace}
\usepackage{tikz}
\usetikzlibrary{arrows.meta, positioning, calc, decorations.markings,
                decorations.pathreplacing, shapes.geometric}
\usepackage{bm}
\usepackage{array}
\usepackage{booktabs}
\usepackage{enumitem}
\usepackage{xcolor}
\usepackage{graphicx}
\usepackage{float}
\usepackage{pgfplots}
\pgfplotsset{compat=1.18}
\hypersetup{colorlinks=true,
            linkcolor=blue!60!black,
            citecolor=blue!60!black,
            urlcolor=blue!60!black,
            pdfauthor={Nicholas G. Polson and Daniel Zantedeschi},
            pdftitle={Kakeya Conjecture and Conditional Kolmogorov Complexity},
            pdfkeywords={algorithmic dimension, conditional complexity, fiber
                         decomposition, Kakeya conjecture, Kolmogorov complexity,
                         point-to-set principle}}

\newtheorem{proposition}{Proposition}

\newtheorem{corollary}{Corollary}
\theoremstyle{definition}
\newtheorem{definition}{Definition}

\theoremstyle{remark}
\newtheorem{remark}{Remark}
\newtheorem{example}{Example}

\newcommand{\R}{\mathbb{R}}
\newcommand{\N}{\mathbb{N}}
\newcommand{\Z}{\mathbb{Z}}
\newcommand{\Sph}{\mathbb{S}}
\newcommand{\eps}{\varepsilon}

\newcommand{\cK}{\mathcal{K}}

\newcommand{\cF}{\mathcal{F}}

\newcommand{\dimH}{\dim_{\mathrm H}}
\newcommand{\dimA}{\dim^{A}}

\newcommand{\restr}{\!\upharpoonright\!}

\title{\textbf{Kakeya Conjecture and Conditional Kolmogorov Complexity}}
\author{Nicholas~G.~Polson%
  \thanks{Booth School of Business, University of Chicago.
          E-mail: \texttt{ngp@chicagobooth.edu}.}%
  \and
  Daniel~Zantedeschi%
  \thanks{Muma College of Business, University of South Florida.
          E-mail: \texttt{danielz@usf.edu}.}}
\date{Working paper, March 2026.  Comments welcome.}
\begin{document}
\maketitle

\begin{abstract}
This paper develops an information-theoretic framework for algorithmic complexity under
regular identifiable fibering.  The central question is: when a decoder is given information
about the fiber label in a fibered geometric set, how much can the residual description length
be reduced, and when does this reduction fail to bring dimension below the ambient rate?
We formulate a directional compression principle, proposing that sets admitting regular,
identifiable fiber decompositions should remain informationally incompressible at ambient dimension,
unless the fiber structure is degenerate or adaptively chosen.  The principle is phrased in the
language of algorithmic dimension and the point-to-set principle of Lutz and Lutz, which translates
pointwise Kolmogorov complexity into Hausdorff dimension.  We prove an exact analytical result:
under effectively bi-Lipschitz and identifiable fibering, the complexity of a point splits additively
as the sum of fiber-label complexity and along-fiber residual complexity, up to logarithmic overhead,
via the chain rule for Kolmogorov complexity.  The Kakeya conjecture (asserting that sets containing
a unit segment in every direction have Hausdorff dimension~$n$) motivates the framework.
The conjecture was recently resolved in $\R^3$ by Wang and Zahl~\cite{wangzahl2025}; it
remains open in dimension $n \geq 4$, precisely because adaptive fiber selection undermines
the naive conditional split in the general case.
We isolate this adaptive-fibering obstruction as the key difficulty and propose a formal research
program connecting geometric measure theory, algorithmic complexity, and information-theoretic
compression.
\end{abstract}


\section{Introduction}
\label{sec:intro}

Information theory has long studied how side information at the decoder
reduces encoding rate.  In the classical Slepian--Wolf and Wyner--Ziv frameworks, a compressor
exploits correlation between a source and the decoder's side information to reduce transmission
rate.  This paper develops a geometric variant: suppose a point belongs to a structured geometric
object organized by fibers (line segments, curves, or submanifolds) indexed by some label.
If the decoder is given side information about the fiber label, how much can the residual
description length be reduced?

More formally, consider $x \in \R^n$ decomposed as $x = \psi(z, u)$, where $z$ is a fiber
label and $u$ is a coordinate along the fiber.  The encoder transmits a finite-precision
description of $x$; the decoder holds side information about $z$.  The residual complexity is
the Kolmogorov complexity of $u$ given $z$.  The compression principle asks: under what
conditions does the total complexity split additively, and when does this additive split force
the point to remain incompressible at the ambient informational dimension?

\subsection{Kakeya as the Canonical Configuration}

The Kakeya conjecture provides the archetypal instance.  A Kakeya set $E \subseteq \R^n$
contains a unit line segment in every direction on $\Sph^{n-1}$.  The conjecture asserts
$\dimH(E) = n$, despite the possibility of measure-zero constructions.  The difficulty is
fundamental: a combinatorial richness condition (a segment in every direction) must force
a metrical conclusion (full Hausdorff dimension).

The algorithmic viewpoint makes the connection transparent.  For a point $x = a(e) + te$
in a Kakeya set, the chain rule decomposes description length at precision $r$ into three
components:
\begin{align*}
K(e \restr r)                          &\approx (n-1)r
  && \text{directions fill } \Sph^{n-1}, \\
K(t \restr r \mid e \restr r)          &\approx r
  && \text{along-segment position, not determined by } e, \\
K(a(e) \restr r \mid e \restr r)       &= O(\log r)
  && \text{base point, under regularity.}
\end{align*}
Combining: $K(x \restr r) \approx (n-1)r + r = nr$, so $\dim(x) = n$, and the
point-to-set principle \cite{lutz2017} then yields $\dimH(E) = n$.

The key subtlety is the regularity of the base point map $e \mapsto a(e)$.  When the
map is Lipschitz (the \emph{sticky Kakeya} condition of Lutz and Stull \cite{lutzstull2020}:
$|a(e)-a(e')| \lesssim |e-e'|$), the third term satisfies $K(a(e)\restr r \mid e\restr r)
= O(\log r)$ and the argument holds.  In the non-sticky case, $K(a(e)\restr r \mid e\restr r)$
can be as large as $O(r)$: the base point is then informationally independent of the
direction, providing an adversarial oracle with an extra dimension to exploit.  From the
present perspective, this identifies a central information-theoretic obstruction in
dimensions $n \geq 4$ (see Remark~\ref{rem:R3}).

\subsection{Algorithmic Dimension as the Unifying Language}

The precise language is provided by algorithmic dimension, defined as the limiting ratio of
Kolmogorov complexity to precision.  The point-to-set principle of Lutz and Lutz bridges
pointwise algorithmic dimension to global Hausdorff dimension: for analytic sets, the
Hausdorff dimension equals the minimax value of oracle-relativized pointwise dimension.
This framework encodes the side-information question: an oracle represents a decoder with
access to geometric side information, and proving full dimension requires showing that no
oracle can simultaneously compress all points of the set below the ambient rate.  At the
pointwise level, the ambient-dimension heuristic is an incompressibility claim, closely
aligned with Martin-L\"of-style algorithmic typicality under admissible side information.

\subsection{Kolmogorov Complexity and Conditional Complexity}

The central quantity throughout the paper is \emph{Kolmogorov complexity}.  For a binary
string $x$, the Kolmogorov complexity $K(x)$ is the length of the shortest program on a
fixed universal prefix-free Turing machine that outputs $x$.  Intuitively, $K(x)$ is the
minimum number of bits needed to describe $x$.

The \emph{conditional Kolmogorov complexity} $K(x \mid y)$ is the length of the shortest
program that outputs $x$ when given $y$ as an auxiliary input.  It measures how much
information about $x$ remains once $y$ is known, the algorithmic analogue of conditional
entropy $H(X \mid Y)$.  The \emph{chain rule} states
\[
K(x, y) = K(y) + K(x \mid y) + O(\log K(x,y)),
\]
mirroring the Shannon identity $H(X, Y) = H(Y) + H(X \mid Y)$.

In this paper, $x$ is a point in $\R^n$ described at finite precision $r$, $z$ is a fiber
label, and $K(x \restr r \mid z \restr r)$ measures the residual description length once
the fiber label is known.  The compression problem is to understand when and by how much
this conditional complexity is reduced relative to $K(x \restr r)$.

\subsection{The Adaptive-Fibering Obstruction}

The compression principle presupposes a well-defined fiber assignment: each point has a
unique label and along-fiber coordinate.  In irregular Kakeya sets, however, a single point
may lie on multiple line segments, so the decomposition is not unique.  An adversarial oracle
can exploit this multiplicity, selecting the fiber label achieving maximum compression.  This
adaptive-fibering phenomenon identifies a central obstruction from the present
information-theoretic perspective.

\subsection{Contributions and Organization}

The contributions are threefold.  First, we formulate a conditional-compression principle for
fibered geometric objects in information-theoretic language.  Second, we prove an exact
analytical result: under effectively bi-Lipschitz and identifiable fibering, the chain rule
yields
$K^A(x \restr r) = K^A(z \restr r) + K^{A,z}(u \restr r) + O(\log r)$,
with two corollaries connecting pointwise dimension to Hausdorff dimension.
Third, we identify the adaptive-fibering obstruction as the fundamental difficulty and
propose a formal research program at the interface of geometric measure theory, algorithmic
complexity, and source coding.

Beyond the Kakeya conjecture itself, this framework addresses a recurring practical
question in geometry-aware coding systems.  Point-cloud compression uses directional
priors; video coding exploits motion vectors; neural compression of high-dimensional data
and retrieval-augmented generation rely on latent fiber labels.  In all these settings the
decoder holds partial geometric side information, and the key question is: how much does
that label reduce residual description length?  When labels arise from a rich but
non-unique fibering, the naive Shannon or chain-rule bound overstates compression gains
precisely because of the adaptive-fibering obstruction identified here.  The Kakeya
problem is therefore not merely a case study but the canonical stress test for this class
of questions.  The minimax formulation $\mathcal{K}_r(E)$ and the Blackwell preorder on
side-information schemes (Section~\ref{sec:IT}) provide a language for principled rate
allocation whenever fiber labels are ambiguous or degraded.

The paper is organized as follows.
Section~\ref{sec:setup} develops the abstract framework.
Section~\ref{sec:results} presents the central proposition, proof, and corollaries.
Section~\ref{sec:kakeya} applies the framework to Kakeya sets.
Section~\ref{sec:obstruction} discusses the adaptive-fibering obstruction.
Section~\ref{sec:IT} bridges the framework to source coding, metric entropy, and Blackwell comparisons.
Section~\ref{sec:practical} illustrates practical implications via worked examples.
Section~\ref{sec:program} outlines a research program of open problems.
Section~\ref{sec:literature} positions the work relative to existing literature.
Section~\ref{sec:conclusion} concludes.
Appendices provide notation, a technical sidebar, and finite-precision schematics.

\section{Problem Setup: Compression with Geometric Side Information}
\label{sec:setup}

\subsection{Abstract Fibered Sets}

Let $X \subseteq \R^n$ be compact and $(Z, d_Z)$ a metric space for fiber labels.
A \emph{fibering} is a family $\{F_z\}_{z \in Z}$ of nonempty compact subsets of $X$
such that $X = \bigcup_{z \in Z} F_z$.  A point $x \in X$ is decomposed as
$x = \psi(z, u)$, where $z \in Z$ is the fiber label and $u$ is a coordinate on $F_z$.
We assume each fiber carries a metric structure so that the along-fiber coordinate
can be described at finite precision.

\subsection{Finite-Precision Descriptions and Conditional Complexity}

Fix precision $r > 0$.  We describe each quantity to within $2^{-r}$ in the Euclidean
metric.  Write $z \restr r$ for a finite-precision encoding of $z$ at resolution $r$,
and similarly for $x \restr r$ and $u \restr r$.

The \emph{conditional complexity} of $x$ given the fiber label $z$ at precision $r$ is
\[
K(x \restr r \mid z \restr r) \approx K(u \restr r \mid z \restr r),
\]
reflecting that once the fiber is identified, the remaining information is the along-fiber
coordinate.  Heuristically, one expects a decomposition of the form
\[
K(x \restr r) \approx K(z \restr r) + K(u \restr r \mid z \restr r) + K(a(z) \restr r \mid z \restr r)
               + O(\log r),
\]
where $K(a(z) \restr r \mid z \restr r)$ is the overhead from specifying the fiber's location.
The exact conditions under which this decomposition holds as an equality up to $O(\log r)$
are the subject of Proposition~\ref{prop:exact}.

\subsection{Identifiability and Ambiguity Gain}

A fibering is \emph{identifiable} at $x$ if $\{z : x \in F_z\}$ is a singleton; it is
\emph{identifiable} on $X$ if identifiable at every point.

The \emph{ambiguity gain} at $x$ is
\[
\Gamma_r(x) := K(x \restr r) - \inf_{z:\, x \in F_z} K(x \restr r \mid z \restr r),
\]
measuring the advantage an adaptive compressor gains by pointwise fiber selection.
In identifiable fiberings with stable parametrization, $\Gamma_r(x)$ is sublinear in $r$.
In irregular regimes, $\Gamma_r(x)$ can be linear in $r$, indicating substantial adaptive
compression.

\section{Exact Results in Regular Fibering Regimes}
\label{sec:results}

\subsection{The Central Proposition}

\begin{proposition}[Effective additive decomposition]
\label{prop:exact}
Let $Z \subseteq \R^{n-1}$ and $U \subseteq \R$ be compact sets, encoded at precision $r$
by standard prefix-free binary descriptions of dyadic approximations.  Let $X \subseteq \R^n$
be compact, and let $\psi: Z \times U \to X$ be given by $\psi(z, u) = a(z) + u \cdot \phi(z)$,
where:
\begin{enumerate}[label=(\roman*)]
\item \emph{Effectively bi-Lipschitz:} there exist computable constants $L_1, L_2 > 0$,
each specifiable in $O(\log r)$ bits at precision $r$, with
\[
L_1 \|(z_1,u_1)-(z_2,u_2)\| \leq \|\psi(z_1,u_1)-\psi(z_2,u_2)\|
\leq L_2 \|(z_1,u_1)-(z_2,u_2)\|
\]
for all $(z_1,u_1),(z_2,u_2)\in Z\times U$;
\item \emph{Identifiable:} $\psi$ is injective on $Z \times U$, so each $x \in X$ has
a unique preimage $(z(x), u(x))$; and
\item \emph{Computable with effective modulus:} $\psi$ is a computable function, meaning
there is a Turing machine that on input $(z\restr r, u\restr r, r)$ outputs
$\psi(z,u)\restr r$ in finitely many steps.  Equivalently, the maps $a\colon Z\to\R^n$
and $\phi\colon Z\to\Sph^{n-1}$ are computable with a computable modulus of
uniform continuity at each precision $r$.
\end{enumerate}
Then for every oracle $A \subseteq \N$ and every $x = \psi(z, u) \in X$,
\[
K^A(x \restr r) = K^A(z \restr r) + K^{A,z}(u \restr r) + O(\log r),
\]
where $K^A$ denotes prefix-free Kolmogorov complexity relative to oracle $A$, $K^{A,z}$
denotes complexity with oracle $A$ augmented by the string $z\restr r$, and the $O(\log r)$
constant depends on $L_1, L_2$ but not on $x$.
\end{proposition}

\begin{proof}
\textbf{Encoding model.}
Throughout, $z\restr r$ denotes the $r$-bit dyadic approximation to $z$ (the nearest
point of $(2^{-r}\Z)^{n-1}$), and similarly for $u\restr r \in 2^{-r}\Z$ and
$x\restr r \in (2^{-r}\Z)^n$.  Kolmogorov complexity is prefix-free and defined relative
to a fixed universal oracle machine $\mathcal U$.

\textbf{Upper bound.}
Let $p_z$ be an optimal $A$-program for $z\restr r$ (so $|p_z|=K^A(z\restr r)$) and $p_u$
an optimal $(A,z\restr r)$-program for $u\restr r$ (so $|p_u|=K^{A,z}(u\restr r)$).
Define program $P_+$:
\begin{enumerate}[label=(\alph*),leftmargin=*,nosep]
\item Decode the self-delimiting pair $\langle p_z, p_u\rangle$.  The prefix-length
      encoding of $|p_z|$ uses $2\lceil\log_2(|p_z|+1)\rceil \leq O(\log r)$ bits,
      since $|p_z|\leq (n-1)r + O(1)$.
\item Run $p_z$ with oracle $A$ to recover $z\restr r$; run $p_u$ with oracle $(A,z\restr r)$
      to recover $u\restr r$.
\item Evaluate $\psi(z\restr r, u\restr r)$.  By the Lipschitz upper bound,
      $\|\psi(z\restr r, u\restr r) - x\| \leq L_2\sqrt{n}\cdot 2^{-r}$,
      so rounding to precision $r+c_+$, where $c_+ = \lceil\log_2(L_2\sqrt{n})\rceil = O(1)$,
      yields a string within $2^{-r}$ of $x$.  Output this string as $x\restr r$.
\end{enumerate}
The total program length is $|p_z|+|p_u|+O(\log r)$, giving
\begin{equation}
K^A(x \restr r) \leq K^A(z \restr r) + K^{A,z}(u \restr r) + O(\log r). \label{eq:ub}
\end{equation}

\textbf{Lower bound.}
We construct an explicit inversion procedure $P_-$ that computes $(z\restr r, u\restr r)$
from $x\restr r$ with $O(\log r)$ description overhead.

Set $c_- = \lceil\log_2(4L_2\sqrt{n}/L_1)\rceil + 1$ (a fixed integer, $O(1)$, determined
by $L_1,L_2$).  Run $P_-$ at grid resolution $s = r + c_-$:
\begin{enumerate}[label=(\alph*),leftmargin=*,nosep]
\item Encode the parameter $r$ as an auxiliary input to $P_-$ at cost $\lceil\log_2 r\rceil=O(\log r)$
      bits.  The constants $L_1, L_2, c_-$ are hardcoded in $P_-$.
\item Enumerate the finite dyadic grid
      $\mathcal G_s = \bigl(Z \cap 2^{-s}\Z^{n-1}\bigr)\times\bigl(U\cap 2^{-s}\Z\bigr)$
      in lexicographic order.  Compactness of $Z\times U$ bounds $|\mathcal G_s|\leq C_0\cdot 2^{ns}$
      for a computable constant $C_0$.
\item For each $(z'',u'')\in\mathcal G_s$, compute $\psi(z'',u'')$ and test whether
      $\|\psi(z'',u'')-x\restr r\|\leq 2\cdot 2^{-r}$.
      Halt at the first match and output $(z''\restr r, u''\restr r)$.
\end{enumerate}

\emph{Correctness.}  Let $(z^\dagger,u^\dagger)\in\mathcal G_s$ be the nearest grid point to
$(z,u)$.  Then $\|(z^\dagger,u^\dagger)-(z,u)\|\leq\sqrt{n}\cdot 2^{-s}$, so by the Lipschitz
upper bound and the triangle inequality,
\[
\|\psi(z^\dagger,u^\dagger)-x\restr r\|
  \leq L_2\sqrt{n}\cdot 2^{-s} + 2^{-r}
  \leq \bigl(L_2\sqrt{n}\cdot 2^{-c_-}+1\bigr)\cdot 2^{-r}
  \leq 2\cdot 2^{-r},
\]
where the last step uses $c_- \geq \log_2(L_2\sqrt{n})$, so the search halts.  For any
matching point $(z'',u'')$, the bi-Lipschitz lower bound gives
\begin{equation}
\|(z'',u'')-(z,u)\|
  \;\leq\; \frac{\|\psi(z'',u'')-x\restr r\| + \|x\restr r - x\|}{L_1}
  \;\leq\; \frac{3\cdot 2^{-r}}{L_1}. \label{eq:match}
\end{equation}

\emph{Exact recovery.}
The bound \eqref{eq:match} controls how far the grid match $(z'',u'')$ can be from the
true $(z,u)$, but the ratio $3/L_1$ may exceed $1$, so $(z''\restr r,u''\restr r)$ need
not equal $(z\restr r,u\restr r)$ in general.  We handle this as follows.  Using floor
truncation, $w\restr r := \lfloor w_i\cdot 2^r\rfloor\cdot 2^{-r}$ componentwise,
the value $z\restr r$ lies in the $\ell^\infty$-ball of radius $\lceil 3/L_1\rceil\cdot 2^{-r}$
centered at $z''\restr r$.  The number of $r$-scale dyadic cells within this ball is at most
$(2\lceil 3/L_1\rceil + 1)^n$, a constant $M$ depending only on $L_1$ and $n$.
The program $P_-$ therefore outputs $(z''\restr r, u''\restr r)$ together with a
$\lceil\log_2 M\rceil$-bit \emph{auxiliary pointer} appended as hardwired advice, not
derived from the first match itself.  The pointer indexes $z\restr r$ within the list of
at most $M$ candidate cells; a companion subroutine checks each candidate against the
original $x\restr r$ via $\psi$ (using hypothesis (iii)) and identifies the unique match.
From this, $(z\restr r, u\restr r)$ is recovered exactly.  The pointer length is $O(1)$
(depending only on $L_1, n$, both hardcoded), so it is absorbed into the $O(\log r)$ budget.

\emph{Overhead accounting.}
$P_-$ is a fixed program with $L_1,L_2,c_-,M$ hardcoded, and the only runtime input is
$r$ (at cost $O(\log r)$ bits), so the inversion contributes
\begin{equation}
K^A(z \restr r,\, u \restr r) \leq K^A(x \restr r) + O(\log r). \label{eq:inv}
\end{equation}
By the chain rule for prefix-free complexity (with $K^A(z\restr r)\leq (n-1)r+O(1)$),
\begin{equation}
K^A(z \restr r)+K^{A,z}(u \restr r)
  \leq K^A(z \restr r,u \restr r)+O\!\bigl(\log K^A(z\restr r)\bigr)
  \leq K^A(x \restr r)+O(\log r). \label{eq:lb}
\end{equation}
Combining \eqref{eq:ub} and \eqref{eq:lb} yields $K^A(x\restr r)=K^A(z\restr r)+K^{A,z}(u\restr r)+O(\log r)$,
where the constant depends on $L_1,L_2,n$ but not on $x$ or $r$.
\end{proof}

\subsection{Corollary: Additivity of Algorithmic Dimension}

\begin{corollary}
\label{cor:dimension-additive}
Under the hypotheses of Proposition~\ref{prop:exact}, for any oracle $A$ and
$x = \psi(z,u) \in X$,
\[
\dimA(x) = \liminf_{r\to\infty} \frac{K^A(z \restr r) + K^{A,z}(u \restr r)}{r}.
\]
(The proposition gives a pointwise $O(\log r)$ comparison; existence of the limit is
not asserted.  If both $\lim_{r\to\infty}K^A(z\restr r)/r$ and
$\lim_{r\to\infty}K^{A,z}(u\restr r)/r$ exist individually, the liminf is a limit
and equals $d_Z + d_F$.)  In particular, if $z$ has algorithmic dimension $d_Z$
relative to $A$ and $u$ has dimension $d_F$ relative to $A$ and $z$, then
$\dimA(x) = d_Z + d_F$.
\end{corollary}

\begin{remark}[Incompressibility and Martin-L\"of randomness]
\label{rem:MLrandom}
The condition $\dim(x) = n$ is a \emph{liminf} statement: $\liminf_{r\to\infty} K(x\restr r)/r = n$.
This is distinct from---and weaker than---the condition $K(x\restr r) \geq nr - O(1)$ for
\emph{all} $r$, which is the Levin--Schnorr characterization of Martin-L\"of randomness.
Martin-L\"of random points have full ambient algorithmic dimension ($\dim(x)=n$), but the
converse need not hold: a point with $\dim(x)=n$ need not be Martin-L\"of random.

In our regular-fibering regime, Proposition~\ref{prop:exact} gives a lower bound
$K^A(x\restr r) \geq (n-1)r + r - O(\log r) = nr - O(\log r)$ for suitable oracle $A$,
which is the $o(r)$-error version consistent with the Martin-L\"of benchmark.
From this perspective, the chain-rule argument is an argument that points in a
regular-fibered Kakeya set are \emph{algorithmically typical} at the ambient rate:
directional richness prevents the set from resembling a lower-dimensional subspace.
The oracle in the point-to-set principle plays the role of an adversary attempting to
make $x$ appear compressible; the regular-fibering case shows this adversary fails.
Whether the stronger Levin--Schnorr bound holds (i.e., whether $K(x\restr r)\geq nr-O(1)$
along a dense subsequence) is a separate question not addressed here.
\end{remark}

\subsection{Corollary: Application to Kakeya Sets}

\begin{corollary}[Kakeya regular-fibering case]
\label{cor:kakeya-regular}
Let $E \subseteq \R^n$ be a Kakeya set with a regular identifiable fibering
$\psi: \Sph^{n-1} \times [0,1] \to E$ given by $\psi(e,t) = a(e) + te$,
where the full map $(e,t) \mapsto a(e) + te$ satisfies the three hypotheses of
Proposition~\ref{prop:exact} (effectively bi-Lipschitz, identifiable, and computable
with effective modulus).
Then for any oracle $A$ and any $x \in E$ such that $e$ has dimension $n-1$
relative to $A$ and $t$ has dimension $1$ relative to $A$ and $e$,
\[
\dimA(x) = n.
\]
This aligns the regular-fibering regime with the point-to-set mechanism underlying
full-dimensional lower bounds.
\end{corollary}

\section{Kakeya as the Canonical Direction-Rich Example}
\label{sec:kakeya}

\subsection{Definition and History}

\begin{definition}
A \emph{Kakeya set} in $\R^n$ is a compact set $E \subseteq \R^n$ containing a unit
line segment in every direction: for each $e \in \Sph^{n-1}$, there exists $a(e) \in \R^n$
such that $\{a(e) + te : t \in [0,1]\} \subseteq E$.
\end{definition}

\begin{definition}[Kakeya conjecture]
Every Kakeya set $E \subseteq \R^n$ satisfies $\dimH(E) = n$.
\end{definition}

Besicovitch's 1928 construction \cite{besicovitch1928} yields Kakeya sets of Lebesgue measure zero in $\R^n$ for
all $n \geq 2$, a remarkable geometric achievement showing that direction-rich sets need not
be metrically thick.  Beyond this construction, Besicovitch developed the theory of linearly
measurable sets and projection results that directly anticipate Marstrand's projection
theorem---a foundational result in geometric measure theory establishing that the projections
of sets of Hausdorff dimension $s > 1$ onto almost every line have positive length.
Yet the Kakeya conjecture asserts such sets cannot be \emph{dimensionally} small.
The conjecture is established in the plane ($n=2$) by Davies (1971).  For $n \geq 3$ it
has long resisted resolution, with lower bounds of $\dimH(E) \geq (2n+2)/3 + \eps_n$ due to
Wolff, Katz--Tao, and Hickman--Rogers--Zhang \cite{hickmanrogerszh2019}.  Dvir (2008)
resolved the finite-field analogue via the polynomial method.  Bourgain and Demeter's
decoupling theorem \cite{bourgaindemeter2015} provided a key harmonic-analytic tool
exploited in subsequent dimension bounds.

\begin{remark}[Resolution in $\R^3$]
\label{rem:R3}
Wang and Zahl~\cite{wangzahl2025} proved the Kakeya conjecture in $\R^3$; a streamlined
account appears in~\cite{guth2025}.  The information-theoretic obstruction
perspective developed here remains central for $n \geq 4$: in higher dimensions the
adaptive-fibering obstruction is unresolved, and the compression-theoretic diagnosis
(no oracle can simultaneously reduce complexity below the ambient rate for all points)
provides a language for what any future proof must rule out.
\end{remark}

\subsection{Geometric Representation and Informational Decomposition}

For $x = a(e) + te$ in a Kakeya set, the three geometric components contribute:

\begin{enumerate}
\item Direction $e \in \Sph^{n-1}$: $K(e \restr r) \approx (n-1)r$ bits.
\item Along-segment coordinate $t \in [0,1]$: $K(t \restr r \mid e \restr r) \approx r$ bits.
\item Basepoint $a(e)$: $K(a(e) \restr r \mid e \restr r) = O(\log r)$ in the Lipschitz-regular
      case, or $O(r)$ when irregular.
\end{enumerate}

Under regularity conditions, the total complexity becomes
\[
K(x \restr r) \approx (n-1)r + r + O(\log r) = nr + O(\log r),
\]
suggesting that points in a Kakeya set retain full ambient complexity.

\subsection{Application of the Central Proposition}

The Kakeya fibering $\psi(e,t) = a(e) + te$ fits Proposition~\ref{prop:exact} whenever
the \emph{full} parametrization $(e,t) \mapsto a(e) + te$ satisfies all three hypotheses:
(i) effectively bi-Lipschitz on $\Sph^{n-1}\times[0,1]$ jointly in $(e,t)$,
(ii) injective (identifiable), and (iii) computable with effective modulus.
Regularity of $e \mapsto a(e)$ alone is not sufficient; the needed condition is on
the joint map $\psi$, and in particular the Lipschitz bound must hold globally
across both the direction and the fiber coordinate.  Under these conditions,
Corollary~\ref{cor:kakeya-regular} yields $\dimA(x) = n$ for every oracle $A$.  The
passage to Hausdorff dimension then uses the point-to-set principle: for every oracle $A$
there exists a point $x \in E$ (specifically one whose direction $e$ and along-fiber
coordinate $t$ are algorithmically random relative to $A$) achieving $\dimA(x) = n$,
so $\dimH(E) = \min_A \sup_{x\in E}\dimA(x) = n$.

\textbf{Key observation.}  In Kakeya sets with irregular basepoint map or non-unique fiber
assignment, these conditions fail.  An adversarial oracle can exploit fiber-choice freedom
to reduce effective description length below the ambient rate.  This is the gap the Kakeya
conjecture highlights.

\begin{figure}[t]
\centering
\begin{tikzpicture}[>=Stealth, scale=0.88]
  \coordinate (base) at (0.6,0.3);
  \coordinate (xpt)  at (2.4,1.22);
  \coordinate (edir) at (3.5,1.8);
  \coordinate (origin) at (0,0);

  \draw[thick] (base) -- ($(base)!1!(edir)$);
  \fill (base) circle (2pt);
  \node[below left, font=\small] at (base) {$a(e)$};

  \fill (xpt) circle (2pt);
  \node[above left, font=\small] at (xpt) {$x = a(e)+te$};

  \draw[decorate, decoration={brace, amplitude=5pt, mirror}]
    ($(base)+(0,-0.18)$) -- ($(xpt)+(0,-0.18)$)
    node[midway, below=5pt, font=\small] {$t$};

  \draw[->, thick, gray] (origin) -- (1.1,0.57);
  \node[above left, font=\small, gray] at (0.55,0.42) {$e$};

  \node[align=left, font=\scriptsize, anchor=west] at (3.8,1.9)
    {$K(e\restr r)\approx(n{-}1)r$};
  \node[align=left, font=\scriptsize, anchor=west] at (3.8,1.2)
    {$K(t\restr r\mid e\restr r)\approx r$};
  \node[align=left, font=\scriptsize, anchor=west] at (3.8,0.4)
    {$K(a(e)\restr r\mid e\restr r)=O(\log r)$};
  \node[align=left, font=\scriptsize, anchor=west] at (3.8,-0.1)
    {\textit{(regular regime)}};
\end{tikzpicture}
\caption{Decomposition of a point on a directional fiber.
Direction $e$ contributes $(n-1)r$ bits, the along-fiber coordinate contributes $r$ bits,
and the basepoint overhead is logarithmic in the regular regime.  Total complexity
matches the ambient dimension.}
\label{fig:decomposition}
\end{figure}
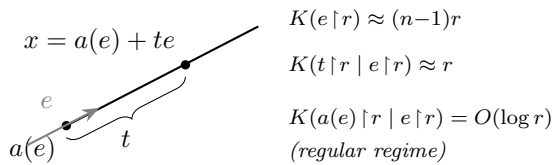

\section{The Obstruction: Ambiguous Fibering and Decoder Advantage}
\label{sec:obstruction}

\subsection{Non-Uniqueness of Decompositions}

In an irregular Kakeya set, a point may lie on multiple segments.  If
$x = a(e_1) + t_1 e_1 = a(e_2) + t_2 e_2$ for distinct $e_1, e_2$, the complexity split
is non-unique.  The decomposition via $e_1$ yields budget
$K(e_1 \restr r) + K(t_1 \restr r \mid e_1 \restr r)$, the decomposition via $e_2$
yields $K(e_2 \restr r) + K(t_2 \restr r \mid e_2 \restr r)$; these can differ
substantially.

A compressor minimizing description length will choose the smallest total.  In the
point-to-set framework, the oracle encodes the entire fiber structure and adaptively
selects the best decomposition for each point at each scale.  The compression principle
requires that \emph{no} such adaptive choice reduces complexity below the ambient rate.
Verifying this uniformly is the central difficulty.

\subsection{Adaptive Side Information and Oracle Advantage}

The point-to-set principle frames dimension as a game: an adversarial oracle chooses
side information to minimize complexity.  In a fibered geometry, the oracle's most
powerful strategy is to select, for each $x$ and precision $r$, the fiber label $e$
maximizing compression gain.  Specifically, the oracle can encode:
\begin{enumerate}
\item the entire fiber structure (basepoints, directions, segment endpoints);
\item an adaptive selection function $\phi_r(x)$ returning the direction minimizing
      $K^A(x \restr r \mid \phi_r(x) \restr r)$ at each scale; and
\item full information about fiber-assignment uniqueness at each point.
\end{enumerate}
This adaptive freedom is the core obstruction.  In the regular-fibering regime of
Proposition~\ref{prop:exact}, the fiber assignment is essentially unique and the
oracle's advantage is neutralized.  In the general case, proving the Kakeya conjecture
would require showing that directional richness imposes a complexity lower bound
no adaptive scheme can circumvent.

\begin{example}[Planar crossing: explicit $\Gamma_r$ calculation]
\label{ex:crossing}
Consider the simplest case in $\R^2$: two line segments crossing at a single point.
Let $e_1 = (1,0)$ (horizontal) with fiber $F_{e_1} = \{(t,\,c) : t\in[0,1]\}$ at height
$c$, and $e_2 = (0,1)$ (vertical) with fiber $F_{e_2} = \{(c',\,s) : s\in[0,1]\}$ at
abscissa $c'$.  Their intersection is $x^* = (c',c)$.

For any $x = (x_1,x_2)$ with $K(x\restr r)\approx 2r$ (both coordinates algorithmically
independent and random at precision $r$):
\begin{itemize}
  \item \emph{Horizontal fiber:} label $z_1 = x_2$,\ residual $u_1 = x_1$;\;
    $K(x\restr r \mid z_1\restr r) \approx K(x_1\restr r \mid x_2\restr r) \approx r$.
  \item \emph{Vertical fiber:} label $z_2 = x_1$,\ residual $u_2 = x_2$;\;
    $K(x\restr r \mid z_2\restr r) \approx r$.
\end{itemize}
The adaptive oracle selects the more compressive label:
\[
\Gamma_r(x)
  = K(x\restr r) - \min\bigl\{K(x\restr r \mid z_1\restr r),\;K(x\restr r \mid z_2\restr r)\bigr\}
  \approx 2r - r = r.
\]
This linear gain $\Gamma_r(x) = r + O(\log r)$ is sharp but not pathological: the chain
rule forces $K(x\restr r \mid z\restr r) \geq K(x\restr r) - K(z\restr r) - O(\log r)
\approx r$ for any direction label $z$ with $K(z\restr r)\leq r$, so no oracle can drive
the residual below $r - O(\log r)$ for a single algorithmically random point.

The obstruction in the full Kakeya setting is therefore \emph{global and oracle-uniform}:
the challenge (Problem~3 in Section~\ref{sec:program}) is to show that a single oracle $A$
cannot simultaneously reduce $K^A(x\restr r \mid z^A(x)\restr r)$ below $r - \Omega(r)$
for \emph{all} $x \in E$ at once, not merely for one isolated intersection.
\end{example}

\subsection{Disintegration of Complexity and the Borel--Kolmogorov Phenomenon}

The core structure of Proposition~\ref{prop:exact} is an \emph{algorithmic disintegration}:
the total description of $x$ splits into the description of its fiber label $z = \pi(x)$
plus the residual description of $x$ given the label.  Algorithmically, this is the
chain rule.  The continuous analogue is the formula for algorithmic dimension under a
projection $\pi$:
\begin{equation}
\label{eq:dimchain}
\dim(x) = \dim(\pi(x)) + \dim(x \mid \pi(x)),
\end{equation}
where $\dim(x \mid \pi(x))$ denotes the dimension of $x$ relative to the oracle $\pi(x)$.
Equation~\eqref{eq:dimchain} mirrors the classical chain rule for Kolmogorov complexity
$K(x,y) = K(y) + K(x \mid y) + O(\log K(x,y))$ and the disintegration theorem in
measure theory, which expresses a joint measure as a base measure on labels integrated
against conditional measures on fibers.  The algorithmic and measure-theoretic statements
are not formally equivalent but share the same logical structure: a complex object is
analyzed by first specifying a coarse label and then describing the residual.

The non-uniqueness of this decomposition has a sharp parallel in classical probability.
The \textbf{Borel--Kolmogorov paradox} demonstrates that conditioning on a measure-zero
event can yield different conditional distributions depending on the sigma-algebra used
to define the limiting procedure.  The conditional distribution of a point on a great
circle, for example, depends on whether one takes the limit through latitude strips,
longitude strips, or some other family of sets.  There is no canonical answer without
specifying the disintegration.

The algorithmic setting makes this identification precise.  \emph{An adaptive fibering
is a choice of sigma-algebra, which is a choice of oracle.}  More concretely: different
projection maps $\pi$ (i.e., different fiber-label functions $x \mapsto z(x)$) yield
different disintegrations of $\dim(x)$ via equation~\eqref{eq:dimchain}, and none is
canonical without additional regularity.  The oracle in the point-to-set principle
encodes exactly this choice.  When one asks ``what is $K^A(x \restr r)$?'', the oracle
$A$ implicitly selects the conditioning structure: which fiber family, which basepoint
map, which disintegration.

The sticky Kakeya condition---$|a(e)-a(e')| \lesssim |e-e'|$---eliminates the
Borel--Kolmogorov freedom by making the fibering essentially unique.  Under stickiness,
all reasonable disintegrations agree to within $O(\log r)$, and the algorithmic chain
rule is stable.  This is precisely why the sticky case is tractable: stickiness
converts a non-canonical conditioning problem into a canonical one.

The full Kakeya conjecture is, in these terms, the statement that no choice of
sigma-algebra---no choice of oracle---can reduce total description length below the
ambient rate.  Unlike the sticky case, the irregular case requires a \emph{uniform
bound across all possible fiberings and all possible oracles simultaneously}.  This
uniformity requirement is why the conjecture resists the techniques that resolve the
sticky case.  A proof is likely to require new tools at the intersection of algorithmic
information theory, harmonic analysis, and geometric measure theory---connecting
oracle-relativized covering arguments with the decoupling inequalities
(cf.\ \cite{bourgaindemeter2015}) that have driven recent harmonic-analytic progress.

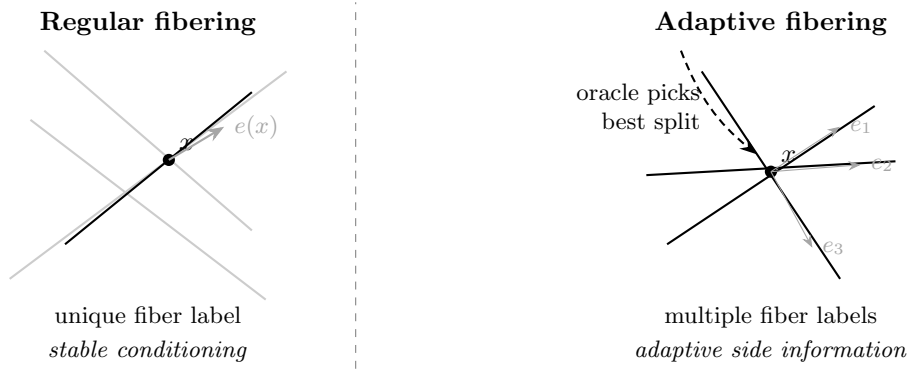
\begin{figure}[t]
\centering
\begin{tikzpicture}[>=Stealth, scale=0.92]
  \begin{scope}[shift={(-4.5,0)}]
    \node[font=\small\bfseries] at (1.5,3.2) {Regular fibering};
    \draw[gray!40, thick] (-0.5,-0.5) -- (3.5,2.5);
    \draw[gray!40, thick] (-0.2,1.8)  -- (3.2,-0.8);
    \draw[gray!40, thick] (0.0,2.8)   -- (3.0,0.2);
    \draw[thick] (0.3,0.0) -- (3.0,2.2);
    \fill (1.8,1.22) circle (2.5pt);
    \node[above right, font=\small] at (1.8,1.22) {$x$};
    \draw[->, thick, gray!70] (1.8,1.22) -- (2.6,1.7);
    \node[right, font=\footnotesize, gray!70] at (2.6,1.7) {$e(x)$};
    \node[align=center, font=\footnotesize] at (1.5,-1.3)
      {unique fiber label\\[2pt]\emph{stable conditioning}};
  \end{scope}
  \begin{scope}[shift={(4.5,0)}]
    \node[font=\small\bfseries] at (1.5,3.2) {Adaptive fibering};
    \draw[thick] (0.0,0.0) -- (3.0,2.0);
    \draw[thick] (-0.3,1.0) -- (3.3,1.2);
    \draw[thick] (0.5,2.5) -- (2.5,-0.5);
    \fill (1.5,1.05) circle (2.5pt);
    \node[above right, font=\small] at (1.5,1.05) {$x$};
    \draw[->, gray!70] (1.5,1.05) -- (2.5,1.7);
    \node[right, font=\footnotesize, gray!70] at (2.5,1.7) {$e_1$};
    \draw[->, gray!70] (1.5,1.05) -- (2.8,1.15);
    \node[right, font=\footnotesize, gray!70] at (2.8,1.15) {$e_2$};
    \draw[->, gray!70] (1.5,1.05) -- (2.1,-0.05);
    \node[right, font=\footnotesize, gray!70] at (2.1,-0.05) {$e_3$};
    \draw[->, thick, densely dashed] (0.2,2.8) to[bend right=15]
      node[midway, left, font=\footnotesize, align=right]
        {oracle picks\\best split} (1.3,1.3);
    \node[align=center, font=\footnotesize] at (1.5,-1.3)
      {multiple fiber labels\\[2pt]\emph{adaptive side information}};
  \end{scope}
  \draw[gray, dashed] (0,-1.8) -- (0,3.5);
\end{tikzpicture}
\caption{Regular vs.\ adaptive fibering.  Left: each point lies on a unique fiber,
yielding a stable conditional-complexity decomposition.  Right: a point lies at the
intersection of multiple fibers; an oracle selects the most compressive representation
pointwise.  This adaptive freedom is the central obstruction.}
\label{fig:adaptive}
\end{figure}

\section{Information-Theoretic Interfaces}
\label{sec:IT}

\subsection{Source Coding with Side Information}

The compression principle translates naturally into source coding.  In Slepian--Wolf and
Wyner--Ziv frameworks \cite{slepianwolf1973,wyziv1976}, a source $X$ is encoded at rate
$R = H(X \mid Z)$ when side information $Z$ is available at the decoder.  Cover
\cite{cover2006} observed that side information is essentially Bayesian conditioning:
providing $Z$ reduces entropy from $H(X)$ to $H(X \mid Z)$.  The present setting is the
geometric and algorithmic analogue: the fiber label $z$ is the conditioning variable, and
$K(x \restr r \mid z \restr r)$ plays the role of $H(X \mid Z)$ for individual descriptions.
The compression principle proposes that this conditional rate cannot drop below the
intrinsic fiber dimension when the fiber structure is regular.  This connection is further
explored in the authors' related work \cite{polsonzantedeschi2025,polsonzantedeschi2025b}.

The parallel is instructive but not exact: source coding concerns i.i.d.\ sequences and
Shannon entropy, while the present setting involves individual strings and Kolmogorov
complexity.

\subsection{Metric Entropy and Covering Complexity}

A second interface connects to metric entropy in the sense of Kolmogorov--Tikhomirov.
The $\eps$-entropy $H_\eps(E)$ measures the logarithm of the minimum number of
$\eps$-balls needed to cover $E$.  For sets of Hausdorff dimension $d$, one expects
$H_\eps(E) \asymp d\log(1/\eps)$.

The compression principle is interpretable as a statement about conditional covering
complexity.  With fiber-label side information, the covering problem reduces to covering
a single fiber $F_e$, with residual entropy $H_\eps(F_e) \asymp d_F \log(1/\eps)$.
The total covering complexity is approximately
\[
H_\eps(E) \approx H_\eps(Z) + \sup_{z \in Z} H_\eps(F_z),
\]
which with $(n-1)$-dimensional directions and one-dimensional fibers yields
$H_\eps(E) \asymp n\log(1/\eps)$, consistent with $\dimH(E) = n$.

\subsection{Blackwell Comparisons, Garbling, and Stitched Fiberings}

A third perspective places admissible side-information schemes in a partial order
analogous to the Blackwell ordering of statistical experiments \cite{blackwell1953}.
Say that $Z_1$ is \emph{more informative} than $Z_2$ for residual compression of $x$ if
\[
K(x \restr r \mid Z_1 \restr r) \;\leq\; K(x \restr r \mid Z_2 \restr r) + o(r)
\]
uniformly over $x \in E$ and all oracles $A$.  This preorder captures which side-information
schemes yield asymptotically smaller residual descriptions, and connects the present
framework to classical comparison-of-experiments theory.  We use ``preorder'' informally
here; a rigorous definition would fix an admissible class of coding models and make
the uniformity conditions explicit.

\medskip\noindent\textbf{Garbling.}\;
A \emph{garbling} of $Z$ is a computable map $G\colon\Sigma^{*}\to\Sigma^{*}$ applied to
the label stream: the decoder receives $G(Z\restr r)$ in place of $Z\restr r$.
Computable maps cannot increase the information content of their input, a fact that
translates directly into a complexity inequality.

\begin{remark}[Garbling weakly increases residual complexity]
\label{rem:garbling}
Let $Z$ be a regular identifiable fiber-label scheme and $G$ a computable garbling.
In the regime of Proposition~\ref{prop:exact},
\[
K^A(x \restr r \mid G(Z) \restr r) \;\geq\; K^A(x \restr r \mid Z \restr r) - O(\log r),
\]
uniformly over $x \in E$ and all oracles $A$.  That is, garbling weakly increases
residual complexity up to logarithmic terms, a data-processing inequality for $K$.
Consequently, $Z$ dominates $G(Z)$ in the informativeness preorder whenever $G$ is not
$O(\log r)$-close to the identity on the relevant label strings.
\end{remark}

\medskip\noindent\textbf{Stitched fiberings.}\;
Global fiber-label schemes are often assembled from local charts.  A
\emph{stitched fibering} consists of a collection $\{\psi_\alpha\}$ of local regular
fiberings on overlapping open sets $\{U_\alpha\}$ covering $E$, with computable transition
functions on overlaps $U_\alpha\cap U_\beta$.  Within each chart, Proposition~\ref{prop:exact}
applies and the complexity split is clean.  On overlaps, however, a point
$x\in U_\alpha\cap U_\beta$ may receive fiber labels $z_\alpha$ and $z_\beta$ that are
not canonically identified, generating the ambiguity gain
$\Gamma_r(x) = K(x\restr r) - \inf_{z\in Z(x)} K(x\restr r\mid z\restr r) \geq 0$.
When the transition functions require more than $O(\log r)$ bits to specify, the stitched
scheme fails to dominate any single-chart scheme, and Blackwell comparisons between
overlapping charts become \emph{incomparable} rather than ordered.  This incomparability
is a signature of the adaptive-fibering obstruction: no canonical global labeling exists,
and the compression gains predicted by the single-chart analysis cannot be realized uniformly.

\subsection{Robust Conditional Compression and Minimax Formulation}

The compression question is fundamentally adversarial.  Define the
\emph{robust conditional complexity}
\[
\cK_r(E) := \min_{\text{fiberings } \cF} \sup_{x \in E} K(x \restr r \mid z_{\cF}(x) \restr r).
\]
A complete definition requires fixing an admissible class of fiberings, for example
computable regular fiberings or stitched computable atlases; the minimax value depends
on this choice.
The compression principle proposes $\cK_r(E) \geq d_F \cdot r - o(r)$ when the fiber
family is sufficiently rich.  This minimax formulation unifies different perspectives:
the oracle-based argument lower-bounds $\cK_r(E)$ via the point-to-set principle,
the source-coding argument interprets it as worst-case conditional rate, and the
metric-entropy argument relates it to conditional covering numbers.

\begin{figure}[t]
\centering
\vspace{4pt}
\begin{tikzpicture}[
  >=Stealth,
  box/.style={draw, rounded corners=4pt, minimum width=3.1cm, minimum height=1.25cm,
              align=center, font=\small, text width=2.9cm, inner sep=5pt},
  arr/.style={->, thick},
  scale=1.0
]
  \node[box] (geom) at (0,0)
    {\textbf{Geometry}\\[3pt]Fibered set $E$\\family $\{F_z\}$};
  \node[box] (enc) at (5.2,0)
    {\textbf{Encoder}\\[3pt]Point $x \in E$\\precision $r$};
  \node[box] (dec) at (10.4,0)
    {\textbf{Decoder / Oracle}\\[3pt]Side info $z\restr r$\\adaptive fibering};
  \node[box] (out) at (5.2,-3.8)
    {\textbf{Outcome}\\[3pt]Residual complexity\\$\dim^A(x) \stackrel{?}{\geq} n$};

  \draw[arr] (geom) -- (enc)
    node[midway, above, font=\footnotesize] {admissible fibers};
  \draw[arr] (enc) -- (dec)
    node[midway, above, font=\footnotesize] {encoded description};
  \draw[arr] (dec) -- (out)
    node[midway, right, font=\footnotesize, align=left]
      {conditional\\compression};
  \draw[arr] (out) -- (geom)
    node[midway, left, font=\footnotesize, align=right]
      {point-to-set\\$\dimH(E) = n$?};
\end{tikzpicture}
\vspace{6pt}
\caption{Minimax framework for geometric compression.  Geometry determines admissible
fiber decompositions; an encoder selects a point at finite precision; an oracle-like
decoder chooses favorable side information; the point-to-set principle converts surviving
pointwise complexity into a Hausdorff dimension lower bound.}
\label{fig:minimax}
\end{figure}
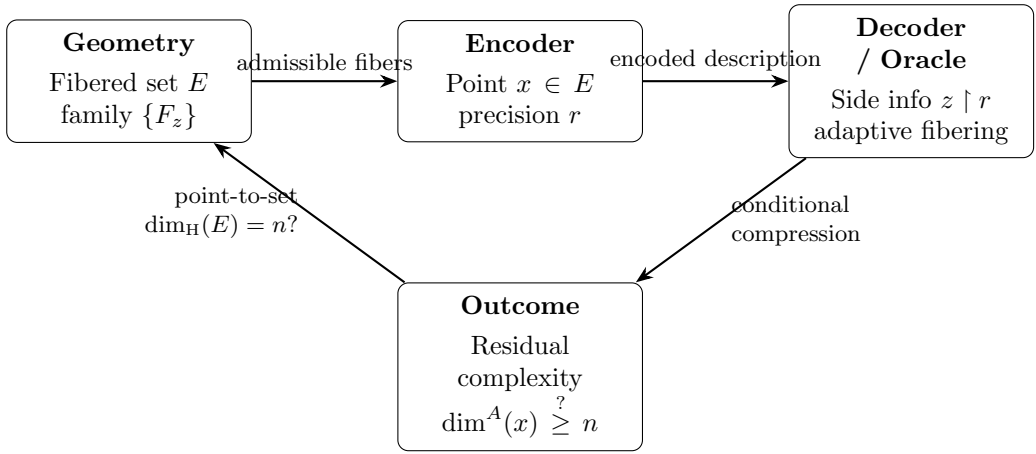

\section{Practical Implications}
\label{sec:practical}

Structured side information is ubiquitous in compression, signal processing, and machine
learning: encoders exploit geometric embeddings, manifold coordinates, or learned latent
factors, identifying a fiber structure and encoding a point by label plus residual.  If
the fiber family is genuinely rich and fibers retain independent variability, side
information about the fiber label cannot reduce residual description length below the
intrinsic fiber dimension.  The distinction between genuine conditional-rate reduction
and ambiguity-enabled artifacts mirrors the distinction between regular and adaptive
fibering.

Garbled or coarsened fiber labels (arising, for example, from quantized latent codes,
noisy cluster assignments, or lossy index compression) act as computable degradations of
the true side-information scheme.  By the data-processing inequality for $K$
(Remark~\ref{rem:garbling}), such garbling can only increase residual complexity, never
reduce it: a compressed or noisy side-information signal provides weakly less benefit than
the original.  Similarly, latent representations assembled by stitching local charts
(as in atlas-based or mixture-of-experts generative models) may exhibit ambiguity at
chart seams.  The stitching obstruction of Section~\ref{sec:IT}
predicts that overlap-induced multiplicity (and any transition functions requiring
super-logarithmic specification) degrades the effective side information and forfeits
the full compression gains available from a single globally identifiable fibering.

Rate allocation in geometry-aware communication systems must answer: how many bits go to
fiber index, how many to along-fiber coordinate, and when is side information genuinely
useful?  The framework, and in particular the Blackwell preorder of
Section~\ref{sec:IT}, provides a language for these questions.

\begin{example}[Blockchain finality as structured side information]
\label{ex:blockchain}
A transaction or local ledger state $x$ in a distributed blockchain system is interpreted
relative to structured contextual labels---chain identity, block depth, checkpoint
status, or finality epoch, which serve as side information $z$.  In a stable finalized
regime, these labels function like an identifiable fiber map: once the relevant chain
history is fixed, the residual uncertainty about the event's status is sharply reduced,
and $K(x\restr r \mid z\restr r) \ll K(x\restr r)$.

Under forks, reorganization risk, or delayed finality, however, the same transaction may
remain compatible with several admissible histories, so the label $z$ is no longer
globally unique.  The admissible label set $\mathcal{Z}(x)$ expands to cover competing
chain branches, and any labeling scheme built by stitching locally coherent node views
creates seam ambiguity at the overlap.  A depth-only label is a garbling of richer
finality information: by Remark~\ref{rem:garbling}, it can only increase residual
complexity.  In a forked regime the competing histories are Blackwell-incomparable
(Section~\ref{sec:IT}): no single label dominates the others in the informativeness
preorder, and apparent compression gains may reflect the choice of history label rather
than genuine reduction in uncertainty.

In the paper's notation, the ledger event plays the role of $x$, the history or finality
label plays the role of $z$, and fork ambiguity enlarges $\mathcal{Z}(x)$ from a
singleton to a multi-element admissible set.  The example illustrates a general lesson:
side information reduces residual description length only when the underlying
representation is stable and globally identifiable: the paper's distinction between
regular and adaptive fibering applies whenever context labels may fail to be globally
unique.
\end{example}

\begin{example}[Retrieval context as structured side information]
\label{ex:retrieval}
In retrieval-augmented systems, an output $x$ (a generated answer, decision, or
ranked result) is interpreted relative to auxiliary context such as retrieved passages,
memory entries, or prompt scaffolds.  These objects act as structured side information
$z$.  A richer retrieval context is naturally more informative than a coarsened summary
or truncated prompt, suggesting a Blackwell-style comparison among context-label schemes
(Section~\ref{sec:IT}); summaries or keyword extracts may be viewed as garblings of
fuller contextual evidence.

The complication is that practical systems assemble context from multiple partially
overlapping sources (retrieved passages, local memories, system prompts, intermediate
summaries) that must be stitched into a single conditioning structure.  When those
pieces do not define a unique global contextual representation, the same output $x$ may
remain compatible with several admissible context labels, enlarging $\mathcal{Z}(x)$.
In such a regime, apparent reductions in residual description length reflect adaptive
context selection rather than genuine informational resolution.  In the notation of the
present paper, the output plays the role of $x$, the retrieved context plays the role
of $z$, and the admissible context family corresponds to $\mathcal{Z}(x)$.

More generally, the value of contextual side information depends not only on how much
content it carries but on whether the conditioning structure it defines is globally
identifiable: the paper's distinction between regular and adaptive fibering applies
naturally to any system whose outputs depend on dynamically assembled context.
\end{example}

\section{A Formal Research Program}
\label{sec:program}

We formulate six problems constituting a research program for the compression-theoretic
study of directional geometries.

\medskip
\noindent\textbf{Problem 1} (Geometric conditional rate function).
Define a rigorous geometric conditional rate function for fibered subsets of $\R^n$,
interpolating between pointwise algorithmic complexity and global metric entropy.

\medskip
\noindent\textbf{Problem 2} (Stability of the conditional split).
Identify sufficient conditions on $\{F_z\}$ and $\psi$ under which
$K(x \restr r) = K(z \restr r) + K(u \restr r \mid z \restr r) + o(r)$
holds uniformly over $x \in E$ and all oracles $A$.

\medskip
\noindent\textbf{Problem 3} (Oracle-uniform lower bounds).
For a fibered set $E$, show that for every oracle $A$,
\[
\sup_{x \in E} \dim^A(x) \geq d_Z + d_F,
\]
without regularity assumptions on the fiber assignment.  This is the core technical
challenge underlying the full compression principle.  Existing techniques (effective
Hausdorff measure, energy methods, Frostman-type arguments) succeed when the fibering is
identifiable because identifiability lets one separate the fiber-label bits from the
along-fiber bits via the chain rule uniformly.  Without identifiability the chain rule
decouples only pointwise, and no current method forces the oracle to ``pay twice''
for the fiber-label overhead across all points simultaneously.  A promising
approach is to track the \emph{effective} Hausdorff content of the fiber-label family
as a function of precision, translating covering-number lower bounds directly into
oracle-relativized complexity bounds.

\medskip
\noindent\textbf{Problem 4} (Algorithmic dimension and metric entropy).
Establish quantitative relationships between pointwise algorithmic dimension and global
metric entropy for geometrically structured families.

\medskip
\noindent\textbf{Problem 5} (Minimax side-information game).
Formalize the minimax game between encoder selecting $x \in E$ and decoder selecting
a fiber assignment.  Characterize the game value for natural classes of fibered sets.

\medskip
\noindent\textbf{Problem 6} (Instantiation in specific geometries).
Investigate whether specific geometric constructions instantiate the compression
principle, including Kakeya configurations with controlled basepoint maps and
extractor-type constructions in additive combinatorics.

\section{Positioning Relative to Existing Literature}
\label{sec:literature}

This paper sits at the intersection of three intellectual traditions.

\emph{Kakeya theory and geometric measure theory.}
The Kakeya conjecture has been studied via harmonic analysis, combinatorial geometry,
and additive combinatorics.  Major contributions include dimension bounds of Wolff,
Katz--Tao, and Hickman--Rogers--Zhang~\cite{hickmanrogerszh2019}, and decoupling
inequalities of Bourgain--Demeter~\cite{bourgaindemeter2015}.  Dvir's finite-field
resolution via the polynomial method is striking.  Wang and
Zahl~\cite{wangzahl2025} recently proved the conjecture in $\R^3$, with a streamlined
account in~\cite{guth2025}.  The present paper does not contribute new geometric
bounds; it proposes an information-theoretic framework for the dimensional question in
arbitrary dimension, where the conjecture remains open.

\emph{Leonid Levin and the complexity-measure connection.}
The bridge between algorithmic randomness and complexity was substantially developed by
Leonid Levin, a student of Kolmogorov at Moscow State University and later at Boston
University.  Levin's work on universal search, resource-bounded complexity, and the
connections between randomness tests and incompressibility helped establish the modern
framework in which measure-theoretic and algorithmic notions of ``typicality'' coincide.
Lutz's constructive Hausdorff dimension and the point-to-set principle are direct
descendants of this program: they translate the local complexity rate of a single point
into global dimensional information about the set containing it.

\emph{Algorithmic dimension and the point-to-set principle.}
The framework of algorithmic dimension (Lutz, Mayordomo, and collaborators) and the
point-to-set principle (J.~Lutz, N.~Lutz) provide our mathematical language.
Geometric applications by Lutz and Stull demonstrate the power of this approach.
What is classical is the framework; what is proposed is the isolation of adaptive-fibering
as a general structural phenomenon.

\emph{Information theory: entropy, coding, and extraction.}
Connections to source coding (Slepian--Wolf, Wyner--Ziv), metric entropy
(Kolmogorov--Tikhomirov), and extractors are suggested but not formalized.
Formalizing these connections is part of the research program.

More broadly, the paper touches two distinct traditions beyond the geometric core:
algorithmic randomness, where Martin-L\"of typicality supplies a natural incompressibility
benchmark, and the foundations of conditioning, where the Borel--Kolmogorov phenomenon
serves as an interpretive reminder that lower-dimensional conditional structure depends
on the chosen mode of disintegration.

In summary, this is a brief analytical and conceptual paper.  It takes an existing
mathematical framework, applies it to a classical geometric problem, and develops a
general compression principle with interfaces to information theory.

\section{Conclusion}
\label{sec:conclusion}

The central idea is straightforward: geometric objects supporting a rich family of
directional fibers should be informationally incompressible.  The Kakeya conjecture is
the most prominent instance.

We have formulated this as a conditional-compression principle in the language of algorithmic
dimension.  The decomposition $K(x \restr r) \approx K(e \restr r) + K(t \restr r \mid e \restr r)$
captures that total point complexity is approximately the sum of directional and along-fiber
complexities.  Under regularity conditions (Lipschitz basepoint maps, identifiable
parametrizations), the decomposition is stable and clarifies the mechanism of known
lower-bound arguments.

The real difficulty is the adaptive-fibering obstruction: when the fiber assignment is
irregular, a point admits multiple decompositions, and an adversarial decoder exploits
ambiguity to reduce description length.  We identify this as the information-theoretic
core of the Kakeya problem.

The contribution is a compression principle and research program, not a resolution of
the Kakeya conjecture.  We outline concrete problems at the interface of geometric measure
theory, algorithmic information theory, and source coding.  Viewed this way, the challenge
is not merely to bound geometric dimension, but to understand which forms of structured
side information truly reduce description length, and which do not.

\appendix

\section{Notation and Terminology}
\label{app:notation}

\begin{table}[!h]
\renewcommand{\arraystretch}{1.2}
\caption{Principal Notation}
\label{tab:notation}
\centering
\footnotesize
\begin{tabular}{>{\raggedright}p{2.6cm} p{5.2cm}}
\toprule
\textbf{Symbol} & \textbf{Meaning} \\
\midrule
$\R^n$ & Euclidean $n$-space \\
$\Sph^{n-1}$ & Unit sphere in $\R^n$ \\
$e \in \Sph^{n-1}$ & Direction (fiber label, Kakeya setting) \\
$a(e)$ & Basepoint of segment in direction $e$ \\
$t \in [0,1]$ & Along-fiber coordinate \\
$x = a(e)+te$ & Point on segment in direction $e$ \\
$z \in Z$ & Fiber label, abstract setting \\
$u$ & Along-fiber coordinate, abstract \\
$F_z$ & Fiber indexed by $z$ \\
$K(\sigma)$ & Prefix-free Kolmogorov complexity \\
$K^A(\sigma)$ & Complexity relative to oracle $A$ \\
$x \restr r$ & Precision-$r$ approximation of $x$ \\
$\dim(x)$ & Algorithmic dimension: $\liminf_{r\to\infty} K(x\restr r)/r$ \\
$\dim^A(x)$ & Oracle-relativized algorithmic dimension \\
$\dimH(E)$ & Hausdorff dimension of $E$ \\
$H_\eps(E)$ & Kolmogorov--Tikhomirov $\eps$-entropy \\
$\cK_r(E)$ & Robust conditional complexity at precision $r$ \\
$\Gamma_r(x)$ & Ambiguity gain: $K(x\restr r)-\inf_z K(x\restr r\mid z\restr r)$ \\
\bottomrule
\end{tabular}
\end{table}

\section{Technical Sidebar: Heuristic Complexity Decomposition}
\label{app:sidebar}

We present a heuristic calculation illustrating the expected complexity decomposition
under favorable conditions.  This is motivational, not a theorem.

\medskip
\noindent\textbf{Setup.}
Let $E \subseteq \R^n$ be a Kakeya set with basepoint map $a: \Sph^{n-1} \to \R^n$.
Fix $x = a(e) + te$ at precision $r$.

\medskip
\noindent\textbf{Chain rule.}
By the chain rule for Kolmogorov complexity,
\[
K(x \restr r) = K(e \restr r) + K(a(e) \restr r \mid e \restr r)
  + K(t \restr r \mid e \restr r, a(e) \restr r) + O(\log r).
\]

\medskip
\noindent\textbf{Directional contribution.}
The direction $e$ on $\Sph^{n-1}$ satisfies $K^A(e \restr r) \geq (n-1)r - o(r)$
for algorithmically random $e$.

\medskip
\noindent\textbf{Basepoint overhead.}
If $a(\cdot)$ is Lipschitz with constant $L$, then
$K(a(e) \restr r \mid e \restr r) = O(\log r)$.
Without Lipschitz regularity, this term can be $O(r)$, the source of the
adaptive-fibering difficulty.

\medskip
\noindent\textbf{Along-fiber contribution.}
For $t$ algorithmically random in $[0,1]$ relative to $(e,A)$:
$K^{A,e}(t \restr r) \geq r - o(r)$.

\medskip
\noindent\textbf{Heuristic conclusion.}
Under the Lipschitz assumption,
\[
K^A(x \restr r) \geq (n-1)r + r - o(r) = nr - o(r),
\]
yielding $\dim^A(x) \geq n$.  The gap to a full theorem lies in controlling basepoint
overhead uniformly without Lipschitz regularity.

\section{Illustrative Finite-Precision Schematic}
\label{app:schematic}

To visualize the distinction between regular and adaptive fibering, we compare
schematic code-length functions.

Define toy code lengths for a point in $\R^n$ at precision $r$:
\begin{align}
L_{\mathrm{reg}}(r) &= nr + c\log(1+r), \\
L_{\mathrm{adapt}}(r) &= nr - \gamma(r),
\end{align}
where $c > 0$ captures basepoint overhead and $\gamma(r) \geq 0$ represents
ambiguity-enabled compression gain.  In the regular regime, code length tracks $nr$
up to logarithmic correction.  Under adaptive side information, $\gamma(r)$ may reduce
code length.

\begin{figure}[t]
\centering
\begin{tikzpicture}
\begin{axis}[
  width=11cm, height=6.5cm,
  xlabel={Precision $r$},
  ylabel={Schematic code length},
  xmin=1, xmax=50,
  ymin=0, ymax=160,
  legend style={at={(0.03,0.97)}, anchor=north west, font=\footnotesize,
                draw=gray!50, fill=white, fill opacity=0.9},
  grid=major,
  grid style={gray!20},
  tick label style={font=\footnotesize},
  label style={font=\small},
  every axis plot/.append style={thick},
]
  \addplot[domain=1:50, samples=80, black, densely dashed] {3*x};
  \addlegendentry{Ambient benchmark $nr$}

  \addplot[domain=1:50, samples=80, black, solid] {3*x + 2*ln(1+x)/ln(2)};
  \addlegendentry{Regular: $nr + c\log(1{+}r)$}

  \addplot[domain=1:50, samples=80, gray, densely dotted] {3*x - sqrt(x)};
  \addlegendentry{Adaptive (mild): $nr - \sqrt{r}$}

  \addplot[domain=1:50, samples=80, gray, solid] {3*x - 0.2*x};
  \addlegendentry{Adaptive (substantial): $nr - 0.2r$}
\end{axis}
\end{tikzpicture}
\caption{Schematic code length under regular and adaptive fiberings ($n=3$).
In the regular regime, code length tracks $nr$ up to logarithmic overhead.
Under adaptive side information, ambiguity-enabled compression gains reduce description
length.  The directional compression principle proposes that for direction-rich geometries
with identifiable structure, gains must be sublinear.}
\label{fig:codelength}
\end{figure}
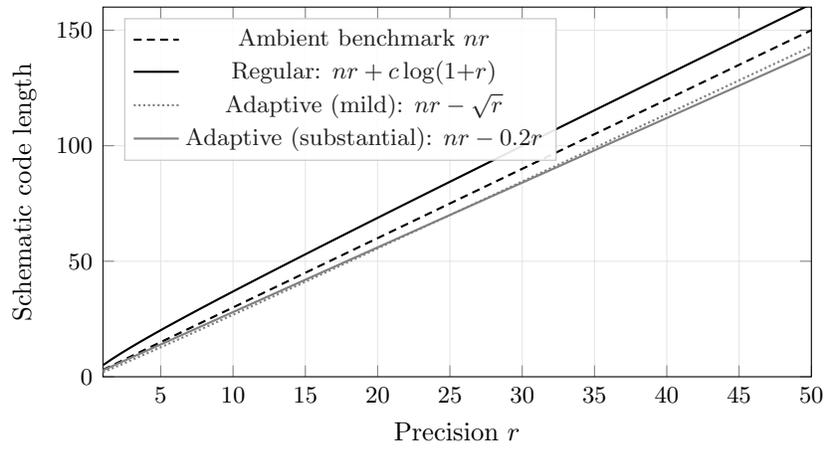

The directional compression principle proposes that for geometries with genuine directional
richness and identifiable fiber structure, the adaptive gain satisfies $\gamma(r) = o(r)$.
The full Kakeya conjecture is that this holds even without identifiability, explaining
why it remains open.



\clearpage
\begin{thebibliography}{10}

\bibitem{lutz2017}
J.~H.~Lutz and N.~Lutz,
``Algorithmic information, plane Kakeya sets, and conditional dimension,''
\textit{ACM Trans.\ Comput.\ Theory}, vol.~10, no.~2, pp.~7:1--7:22, 2018.

\bibitem{lutzstull2020}
J.~H.~Lutz and D.~M.~Stull,
``Bounding the dimension of points on a line,''
\textit{Inform.\ Comput.}, vol.~275, p.~104601, 2020.

\bibitem{davies1971}
R.~O.~Davies,
``Some remarks on the Kakeya problem,''
\textit{Proc.\ Cambridge Philos.\ Soc.}, vol.~69, pp.~417--421, 1971.

\bibitem{dvir2009}
Z.~Dvir,
``On the size of Kakeya sets in finite fields,''
\textit{J.\ Amer.\ Math.\ Soc.}, vol.~22, no.~4, pp.~1093--1097, 2009.

\bibitem{wolff1995}
T.~Wolff,
``An improved bound for Kakeya type maximal functions,''
\textit{Rev.\ Mat.\ Iberoam.}, vol.~11, no.~3, pp.~651--674, 1995.

\bibitem{katztao2002}
N.~H.~Katz and T.~Tao,
``New bounds for Kakeya problems,''
\textit{J.\ Anal.\ Math.}, vol.~87, pp.~231--263, 2002.

\bibitem{slepianwolf1973}
D.~Slepian and J.~K.~Wolf,
``Noiseless coding of correlated information sources,''
\textit{IEEE Trans.\ Inf.\ Theory}, vol.~19, no.~4, pp.~471--480, Jul.\ 1973.

\bibitem{wyziv1976}
A.~D.~Wyner and J.~Ziv,
``The rate-distortion function for source coding with side information at the decoder,''
\textit{IEEE Trans.\ Inf.\ Theory}, vol.~22, no.~1, pp.~1--10, Jan.\ 1976.

\bibitem{lutz2003}
J.~H.~Lutz,
``The dimensions of individual strings and sequences,''
\textit{Inform.\ Comput.}, vol.~187, no.~1, pp.~49--79, 2003.

\bibitem{blackwell1953}
D.~Blackwell,
``Equivalent comparisons of experiments,''
\textit{Ann.\ Math.\ Stat.}, vol.~24, no.~2, pp.~265--272, 1953.

\bibitem{cover2006}
T.~M.~Cover and J.~A.~Thomas,
\textit{Elements of Information Theory}, 2nd~ed.
Hoboken, NJ, USA: Wiley, 2006.

\bibitem{polsonzantedeschi2025}
N.~G.~Polson and D.~Zantedeschi,
``De Finetti + Sanov = Bayes,''
arXiv:2509.13283, Sep.\ 2025.

\bibitem{polsonzantedeschi2025b}
N.~G.~Polson and D.~Zantedeschi,
``Bayesian prediction under moment conditioning,''
arXiv:2510.20742, Oct.\ 2025.

\bibitem{koltikhomirov1959}
A.~N.~Kolmogorov and V.~M.~Tikhomirov,
``$\eps$-entropy and $\eps$-capacity of sets in function spaces,''
\textit{Amer.\ Math.\ Soc.\ Transl.}, vol.~17, pp.~277--364, 1961.

\bibitem{bourgaindemeter2015}
J.~Bourgain and C.~Demeter,
``The proof of the $\ell^2$ decoupling conjecture,''
\textit{Ann.\ of Math.}, vol.~182, no.~1, pp.~351--389, 2015.

\bibitem{besicovitch1928}
A.~S.~Besicovitch,
``On Kakeya's problem and a similar one,''
\textit{Math.\ Z.}, vol.~27, no.~1, pp.~312--320, 1928.

\bibitem{hickmanrogerszh2019}
J.~Hickman, K.~M.~Rogers, and R.~Zhang,
``Improved bounds for the Kakeya maximal conjecture in higher dimensions,''
\textit{Amer.\ J.\ Math.}, vol.~144, no.~6, pp.~1511--1560, 2022.

\bibitem{wangzahl2025}
H.~Wang and J.~Zahl,
``Volume estimates for unions of convex sets, and the Kakeya set conjecture in three dimensions,''
arXiv:2502.17655, Feb.\ 2025.

\bibitem{guth2025}
L.~Guth,
``Outline of the Wang--Zahl proof of the Kakeya conjecture in $\mathbb{R}^3$,''
arXiv:2508.05475, Aug.\ 2025.

\end{thebibliography}
\end{document}